\documentclass{llncs}

\usepackage{amsmath}
\usepackage{amssymb}
\usepackage{tikz}
\usepackage{listings}
\usepackage{algorithm,algpseudocode}
\usetikzlibrary{arrows}
\usetikzlibrary{positioning}
\usetikzlibrary{patterns}
\usetikzlibrary{shapes}

\usepackage{wcet-macros}
\usepackage{comment-macros}
\usepackage{polyhedra-macros}

\newtheorem{ex}{Example}

\title{Abstract Interpretation of Binary Code with Memory Accesses
  using Polyhedra}

\author{Cl\'ement Ballabriga and Julien Forget and Giuseppe Lipari}

\institute{
  University of Lille, CNRS, Centrale Lille, UMR 9189 -- CRIStAL\\
  \email{\{name.surname\}@univ-lille1.fr}
}

\begin{document}

\maketitle

\begin{abstract}
  In this paper\footnote{An earlier version of this paper has been
    submitted to TACAS 2018
    (\url{http://www.etaps.org/index.php/2018/tacas}) for
    peer-review. Compared to the submitted paper, this version
    contains more up-to-date benchmarks in Section 6.} we propose a
  novel methodology for static analysis of binary code using abstract
  interpretation. We use an abstract domain based on polyhedra and two
  mapping functions that associate polyhedra variables with registers
  and memory.

  We demonstrate our methodology to the problem of computing upper
  bounds to loop iterations in the code. This problem is particularly
  important in the domain of Worst-Case Execution Time (WCET) analysis
  of safety-critical real-time code. However, our approach is general
  and it can applied to other static analysis problems.
\end{abstract}

\section{Introduction}


In real-time systems it is important to compute upper bounds to the
execution times of every function, and check that they complete before
their deadlines under all possible conditions. Worst-Case Execution
Time (WCET) analysis consists in computing (an upper bound to) the
longest path in the code. WCET analysis is usually performed on the
binary code, because it needs information on the low-level
instructions executed by the hardware processor in order to compute
the execution time.

In this paper, we propose a static analysis of binary code based on
abstract interpretation using polyhedra. Our motivation is the need to
enhance existing WCET analysis by improving the computation of
upper bounds on the number of iterations in loops, and by detecting
unfeasible paths.

Most analyses by abstract interpretation proposed in the literature are
performed on source code. However, there are several important
advantages in performing static analysis of binary code:
1) we analyze the code that actually runs on the machine, hence 
  no need for additional assumptions on how the compiler works;
2) by gaining access to the memory layout, we can precisely
   identify problems with pointers (alias, buffer overflows, etc.)
   that are not easily identified when working only on source
   code;
3) We can perform the analysis even without access to
  the source code.

To the best of our knowledge, no other work has tackled the problem
of static analysis of binary code using abstract interpretation with
polyhedra. The main reason is probably the difficulty in representing
the state of the program.

In fact, binary code lacks important
structural information: the type of variables, their
structure and relations, their scope and lifetime, etc.
More specifically, in most of the existing abstract interpretation
papers in the literature, the abstract state of a program is represented
by constraints on the \emph{program variables}. The underlying
assumption is that variables are well identified and in a relatively
small number. In the binary code, the notion of program variable is
lost, so we can only analyse processor registers and memory
locations. Unfortunately, a representation of the abstract state that
includes all registers and all possible memory locations is too large to
be managed easily in the analysis.

\paragraph{Contributions of this paper.}

A key observation is that the number of memory locations that are
effectively used in the binary code approximately corresponds to the
number of program variables in the source code (at least, it is of the
same order of magnitude). Based on this observation, in this paper we
propose to identify the subset of registers and memory locations to be
represented in the abstract state 
as the analysis progresses.

To this end, we propose an abstract domain consisting of 1) \emph{a
  polyhedron}, to represent the linear constraints on the polyhedra
variables; 2) a \emph{register mapping} that maps register names to
polyhedra variables; 3) a \emph{memory mapping} that maps addresses
and their values to the respective polyhedra variables. The abstract
state is constructed and modified as the analysis
progresses: in particular, new variables and constraints may be added
to the polyhedra as new memory locations are discovered in the code,
and the two mappings are modified accordingly. We must check that our
abstract representation remains consistent at all stages, for example
when joining the states from two different paths in the code. To do
this, we carefully consider \emph{aliasing} (or \emph{equivalence})
between variables. We formally prove that the operations that
manipulate the abstract state are semantically consistent with the
abstract interpretation framework.

We present the application of this method to the problem of computing
upper bounds to loop iterations. Finally, we evaluate the
performance of our method on benchmarks and we compare with existing
static analysis tools to evaluate the precision of our approach.


\section{Related works}
\label{sec:related}

Many papers on static analysis are based on abstract
interpretation. Approaches vary widely depending on the abstract domain,
the type of target programs (source or binary code, language,
application type, etc.), and the analysis goal (i.e. the kind of
information we want to discover). Here we attempt to summarise the
papers that are the most closely related to our research.


Abstract interpretation using polyhedra has been first described
in~\cite{cousot1978automatic}. It has been used extensively in the
context of compilers. For instance, the PAGAI~\cite{henry2012pagai}
analyzer processes LLVM Intermediate Representation (IR) using various
abstract domains (including polyhedra) to detect several properties
(such as loop invariants). Compared to our approach, LLVM IR is closer
to the source code, as it contains information on variables and their
types.

An important problem when dealing with binary code analysis is to
figure out the set of interesting data locations used by the
program. This is related to pointer analysis (the so-called
aliasing problem), and has been extensively
studied~\cite{Hind:2001:PAH:379605.379665,hardekopf2007ant}. While the
majority of pointer analyses have been proposed in the context of
compiler optimizations, a certain number of ideas can be borrowed and
applied to binary code analysis.

An example of binary code analysis by abstract interpretation is
described in~\cite{balakrishnan2004analyzing} and later improved
in~\cite{reps2008improved}. Compared to our approach, this method uses a
different abstract domain (interval analysis with congruence, and affine
relations~\cite{muller2004precise}). Furthermore, the set of abstract
memory locations are computed by a preprocessing analysis using IDA
Pro~\cite{eagle2011ida}: in contrast, we determine them dynamically
during the analysis.

In this paper, our approach is applied to static loop bound
estimation, in the context of WCET analysis, so we compare our results
with other loop bound estimation tools. The \orange\
tool~\cite{bonenfant2008orange} is based on an abstract interpretation
method defined in~\cite{ammarguellat1990automatic}. It provides a very
fast estimation of loop bounds, but it is restricted to C source code.
SWEET~\cite{sweet} features a loop bound estimator, which works on an
intermediary representation (ALF format). The approach is based on
slicing and abstract interpretation and it generally provides very
tight loop bounds even in complex cases, but the running time of the
analysis appears to depend on the loop bounds, and in our experience
for large loop bounds the analysis did not terminate.

Compared to these existing works, our approach combines the polyhedral
domain with binary code analysis, taking into account memory
accesses; our method is sound and always terminates. 


\section{Abstract domain}
\label{sec:domain}

Abstract interpretation \cite{cousot1992} is a static program analysis
that provides a sound approximation of the semantics of the analyzed
program. Instead of computing exact \emph{concrete program states}
(e.g. data valuation), abstract interpretation computes approximate
\emph{abstract program states} (e.g. inequalities on data valuation).
The set of abstract states is called the \emph{abstract domain}. Our
analysis is based on the \emph{polyhedral abstract domain}, to which
we add information to track relations between polyhedra variables and
registers or memory locations.

\subsection{Polyhedra} 

A \emph{polyhedron} $\mathcal{P}$ denotes a set of points in a
$\mathbb{Z}$ vector space bounded by linear constraints (equalities or
inequalities). More formally, let $C_n$ be the set of linear
constraints in $\mathbb{Z}^n$. Then $\cons{c_1, c_2, ..., c_m}$ (with
$c_i\in C_n$ for $1\leq i\leq m$) denotes the polyhedron consisting of
all the vectors in $\mathbb{Z}^N$ that satisfy constraints $c_1$,
$c_2$, $\ldots$ $c_m$. In our work, we only consider non-strict
inequalities and equalities. 

Like any other abstract domain, the polyhedral domain is equipped with
operations that form a \emph{lattice}~\cite{bagnara2008parma}. We
present these operations below, along with some other classic polyhedra
operations~\cite{bagnara2008parma} that we will use in the
following. Let $p$, $p^{\prime}$ be two polyhedra:
\begin{itemize}
\item We denote $p \subseteq p^{\prime}$ iff
  $\forall s \in p, s \in p^{\prime}$. This is the partial
  order of the lattice;
\item We denote $p^{\prime\prime} = p \cup p^{\prime}$ the smallest
  polyhedron such that $p \subseteq p^{\prime\prime}$ and
  $p^{\prime} \subseteq p^{\prime\prime}$. This operation is the
  \emph{convex hull}. It is the least upper bound of the lattice;
\item We denote $p^{\prime\prime} = p \cap p^{\prime}$ the intersection
  of $p$ and $p^{\prime}$ (i.e. the union of the constraints of $p$ and
  $p^{\prime}$). It is the greatest lower bound of the lattice;
\item The bottom state $\bot$ of the lattice is the empty polyhedron
  (the set of constraints has no solution\footnote{Checking emptiness
    of a polyhedron over $\mathbb{Z}$ is a complex
    operation. Therefore, in practice we approximate it by checking
    emptiness over the rationals.}) and the top state $\top$ is the
  polyhedron containing all points (the set of constraints is empty);
\item We let $\vars{p}$ denote the set of variables appearing in the
  constraints of $p$;
\item Let $|S|$ denote the cardinality of set $S$. We let
$\proj{p}{x_1\ldots x_k}$ denote the projection of polyhedron $p$ on
  space $x_1\ldots x_k$, with $k<|\vars{p}|$ (this effectively removes
  other variables from the polyhedron constraints);
\item We denote $\max(p,x)$ the greatest value of $x$ satisfying the
  constraints of $p$.
\end{itemize}

\subsection{Registers and memory}

In polyhedral analysis of source code, variables of the polyhedra are
related to variables of the source code. In our case, polyhedra
variables are related to registers and memory locations. An abstract
state is defined as a triple $(p,m,\derefer)$, where $p$ is a
polyhedron, $m$ is a \emph{register mapping} and $\derefer$ is an
\emph{address mapping}. In the rest of the paper, the term variable
refers to polyhedron variables.

Let us first consider registers. Let $r$ be a register, $v$ be a
variable of $\vars{p}$ and $m$ be a register mapping. Then we have
$m(r)=v$ iff $v$ represents the value of $r$ in $p$. We denote
$\varsR{p}$ the image of mapping $m$. We denote $m[r:v]$ the mapping
$m^{\prime}$ such that $m^{\prime}(r)=v$ and for every register
$r^{\prime}\neq r$, $m^{\prime}(r^{\prime})=m(r^{\prime})$. In other
words, $m[r:v]$ denotes a single mapping substitution.

Let us now consider memory locations. We want to keep track in our
abstract state of the fact that, say variable $x_0$ represents a memory
address, and the value at this address is represented by variable
$x_1$. To to this, we introduce two more subsets $\varsA{p}$ and
$\varsC{p}$ of $\vars{p}$. Variables in $\varsA{p}$ define memory
addresses, while variables in $\varsC{p}$ define memory values. The
subsets $\varsR{p}$, $\varsA{p}$ and $\varsC{p}$ are disjoint, so we can
distinguish between variables corresponding to memory values, memory
addresses, registers or other. $\derefer$ is a partial bijection (some
addresses may have no corresponding values) from $\varsA{p}$ to
$\varsC{p}$, such that $\derefer(x_1) = x_2$ iff $x_2$ represents the
value at the memory address represented by $x_1$. The substitution
$\derefer[x_1:x_2]$ is defined similarly to $m[r:v]$.

\subsection{Aliasing}

In a general sense, \emph{aliasing} occurs in a program when a data
location can be accessed through several symbolic names. In our context,
we define an \emph{aliasing relation} between two variables $x_1$ and
$x_2$ of a polyhedron $p$ as follows:
\begin{itemize}
\item Must alias or \emph{equivalent}: the constraint $x_1 = x_2$ is
  true for every point in $p$. This is also denoted $x_1\equiv
  x_2$. This holds iff $\cons{x_1=x_2})\subseteq p$;
\item May alias or \emph{overlapping}: the constraint $x_1 = x_2$ is
  true for at least one point in $p$. This holds iff $\cons{x_1=x_2})\cap p\neq\emptyset$;
\item Cannot alias or \emph{independent}: the constraint $x_1 = x_2$
  is false for every point in $p$. This holds iff $\cons{x_1=x_2})\cap p=\emptyset$.
\end{itemize}
The aliasing relation between a register $r$ and a variable $x$ is
defined by the aliasing relation between $m(r)$ and $x$. Similarly, the
aliasing relation between two registers $r_1$, $r_2$ is defined by the aliasing
relation between $m(r_1)$ and $m(r_2)$.

As we will see in the following sections, while progressing in the
analysis, it may happen that our algorithm generates equivalent
variables to represent the same memory location. The presence of
aliasing complicates the analysis, so we make sure that each memory
location is represented by one single variable.


\begin{definition}
  Let $s=(p,m,\derefer)$ be an abstract state. We say that $s$ is a
  \emph{consistent state} iff:
  \[
     \forall \text{distinct } x_1, x_2 \in \varsA{p}, x_1 \not\equiv x_2 
  \]
\end{definition}

In a consistent state, each data location is defined by a single
variable. Indeed, each register is defined by a single variable because
mapping $m$ is a function. The same is true for values contained in
memory locations because $\derefer$ is a function and because the state
is consistent.

To preserve the consistency of abstract states, if at some point in the
analysis we detect equivalent address variables, we merge the variables
and their constraints.  Let $x_1, x_2$ be two distinct variables of
$\varsA{p}$ such that $x_1 \equiv x_2$. We merge $x_1$ and $x_2$ as
follows:
\begin{align*}
      Merge &((p, m, \derefer), x_1, x_2) = (p_1 \cup p_2, m, \derefer^{\prime}) \\
      &p_1 = p \text{ where $\derefer(x_2)$ is replaced by $v_2$} \\
      &p_2 = p \text{ where $\derefer(x_1)$ is replaced by $v_1$, $x_2$ by $x_1$, and $\derefer(x_2)$ by $\derefer(x_1)$} \\
      &\derefer^{\prime} = \derefer[x_2:v_3] \\
      &v_1=\FV, v_2=\FV, v_3=\FV
\end{align*}

Let us describe the operation in detail. We first transform $p$ into two
different polyhedra $p_1$ and $p_2$. Polyhedron $p_1$ is the same as
$p$, except that we eliminate the constraints on $\derefer(x_2)$: this
is done by replacing $\derefer(x_2)$ with a new variable
$v_2$ in all constraints. Polyhedron $p_2$ is the same as $p$, except
that: 1) we remove the constraints on $\derefer(x_1)$; 2) we rename
$x_2$ as $x_1$ and $\derefer(x_2)$ as $\derefer(x_1)$ in all
constraints.  The polyhedron resulting from the merge is the
convex-hull of $p_1$ and $p_2$. Also, $x_2$ is not needed
anymore, so we can change its $\derefer$ mapping to a new variable
$v_3$.

Equivalent address variables are merged whenever adding a new
constraint, using the following unification function:
\[
Unify(p, m, \derefer) =
\begin{cases}
  Merge((p, m, \derefer), x_1, x_2) &\text{ if }\exists x_1,x_2 \in
    \varsA{p}, x_1 \equiv x_2\\
  (p, m, \derefer) &\text{ otherwise.}
\end{cases}
\]


\section{Computing abstract states}
\label{sec:interpretation}

A program is represented by a graph $G = <I, E>$. The set of nodes $I$
is the set of instructions of the program. A directed edge
$(b_1,b_2)\in E$ (where $E\subseteq I\times I$), represents a valid
succession of two basic instructions in the program execution. 

We say that a node $b_i \in I$ is a \emph{predecessor} of $b_j$, and
denote $b_i\pred b_j$, iff $(b_i, b_j)\in E$. We say that $b_i$
\emph{dominates} $b_j$, and denote $b_i\dom b_j$, iff all paths from
the entry node to $b_j$ go through $b_i$. We say that node $h$ is a
\emph{loop header} if it has at least one predecessor $b_i$ such that
$h \dom b_i$. We denote $l_h$ the loop associated to header $h$. An
edge $(b_i,h)$ such that $h\dom b_i$ is called a \emph{back-edge} of
loop $l_h$; any other edge entering the loop header $h$ is an
\emph{entry-edge}.

We now describe our model of the processor architecture. We assume
that all data locations have the same size and that memory accesses
are aligned to the word size. We also assume that function calls are
inlined. We consider a simplified instruction set made up of the
following instructions\footnote{Though this instruction set is very
  small, the principles of our analysis can be  easily extended to
  support a richer set of instructions.}. Let $r_1$, $r_2$, $r_3$ be
registers and $c$ be an integer constant:
\begin{itemize}
\item \lstinline!OP r1 r2 r3!: stores the result of operation
  $OP(r_2,r_3)$ in register $r_1$, where $OP$ denotes an arbitrary
  binary arithmetic or logic operation;
\item \lstinline!BOP r1 r2!: branches to the address contained in $r_1$
  if condition $OP(r_2)$ is true, where $OP$ denotes an arbitrary
  unary logic operation;
\item \lstinline!LOADI r1 c!: loads constant $c$ in register $r_1$;
\item \lstinline!LOAD r1 r2!: loads the value contained in the address
  designated by $r_2$ in register $r_1$;
\item \lstinline!STORE r1 r2!: stores the value of register $r_2$ at
  the address designated by $r_1$.
\end{itemize}

The abstract interpretation of a program consists in computing an
abstract state for each edge of the program. For each node of the
program, the analyses computes the state of the output edge(s) based
on the state of the input edge(s). The state for the entry edge of the
program is $(\top,\emptyset,\emptyset)$. The operations involved in
the computation of abstract states are:
\begin{itemize}
\item $Update$: this is a monotonic function. It takes an
  instruction, the abstract state before the instruction, and returns
  the abstract state after it;
\item $\join$: when an instruction has several input edges (due to
  branching) and is not a loop header, its input state is
  computed by the $\join$ operation;
\item $\widen$: when an instruction is a loop header, its input state
  is computed by the widening operation $\widen$.
\end{itemize}

These operations are applied repeatedly on $G$ until a fixpoint is
reached. Since the polyhedral domain admits infinite ascending chains,
specific mechanisms are used to enforce the convergence to a
fixpoint: this is ensured by the widening operation. A widening
operation for the polyhedral domain has first been proposed
in~\cite{polka97}, and further improved
in~\cite{bagnara2003,simon2006}. In our work, we use 
the widening operator of~\cite{bagnara2003}.

Finally, function $\FV$ returns a new fresh variable that has never
been used at any other point during the analysis. It is very important
that the variable is fresh globally (for the whole analysis) and not
only locally (for the current state), so as to avoid using the same
variable in two different states for representing unrelated
constraints.

Figure \ref{fig:example} reports an example that will be used in the
rest of the section to illustrate the operations on the abstract
states.

  \begin{figure}
    \begin{multicols}{2}
      {\scriptsize\begin{algorithmic}[1] \State LOADI R1 4 \State
          LOADI R2 5 \State LOADI R3 1000 \State LOADI R4 9 \State
          STORE R3 R1 \State EQ R5 R1 R2~~~~\# R5$\gets$R1==R2 \State
          BNZ R4 R5~~~~~~~\# Branch to 9 if R5$\neq$0 \State STORE R3
          R2 \State LOAD R6 R3
        \end{algorithmic}}
    \end{multicols}

  \noindent{\small $\begin{array}{|c|c|c|c|}
      \hline
      Edge & Polyhedron & Registers & Memory \\
      \hline
      (L6,L7) &
                p_1=\langle \{x_1=4,x_2=5,x_3=1000, & m_1=\{R_1:x_1,
                                    & \derefer_1=\{x_5:x_6\}\\
      & x_4=9,x_5=x_3,x_6=x_1, & R_2:x_2,R_3:x_3, & \\
           & x_7=(x_1-x_2)\rangle\} & R_4:x_4,R_5:x_7\} & \\
      \hline
      e=(L8,L9) & p_1\cap\cons{x_7=0,x_8=x_2} & m_1 & \derefer_1[x_5:x_8]\\
      \hline
      e'=(L7,L9) & p_1 & m_1 & \derefer_1\\
      \hline
      join(e,e') & p_3=(p_1\cap\cons{x_7=0,x_6=x_2})\cup p_1=& m_1 & \derefer_2=\{x_5:x_6\}\\
           & \cons{x_1=4,\ldots,x_1\leq x_6\leq x_2} & & \\
      \hline
      (L9,exit) & p_3\cap\cons{x_{10}=x_6} & m_1[R6:x_{10}] &
                                                               \derefer_2\\
      \hline
                    \end{array}$}
                  
                  \caption{Example of analysis}
                  \label{fig:example}
                \end{figure}

\subsection{Binary operation}

We distinguish two possible cases. In the first case, the relation
$r_1 = OP(r_2,r_3)$ is linear and 
the $\update$ function can be defined as follows:
\begin{align*}
	Update & (OP~r_1~r_2~r_3, (p, m,\derefer)) =
                (p^{\prime}, m^{\prime},\derefer) \\
               & p^{\prime} = p \cap \cons{x_i = OP(m(r_2),m(r_3))} \\
               & m^{\prime} = m[{r_1}:x_i] \quad\quad x_i=\FV
\end{align*}
For instance, 
in Figure~\ref{fig:example} 
Line~6 introduces the constraint
$x_7=(x_1-x_2)$ (the equality test expressed as a linear constraint) and
$m_1(R_5)=x_7$.

If no linear relation can be determined, we assume that $r_1$
can contain any value after the update. In that case, the $\update$
function is defined as follows:
\begin{align*}
	Update^{\prime}&(OP~r_1~r_2~r_3, (p, m,\derefer)) = (p, m^{\prime},\derefer) \\
                       & m^{\prime} = m[{r_1}:\FV]
\end{align*}

\subsection{Branching}

Branching instructions have two out-edges, to which different abstract
states can be associated. We use \emph{filtering} to represent the
impact of the condition on the abstract state. It is applied only if the
branching condition is a linear constraint, otherwise the branching
condition is ignored. Let $taken$ denote the edge corresponding to the
case where the branch condition is true, and $not\_taken$
denote the other edge. The abstract states for these edges are computed
as:
\begin{gather*}
  Update_{taken}(BOP~r_1~r_2, (p,m,\derefer) = (p^{\prime},m,\derefer)
  \\
  p^{\prime} = p \cap \cons{BOP(m(r_2))}\\
  Update_{not\_taken}(BOP~r_1~r_2, (p,m,\derefer) = (p^{\prime},m,\derefer)
  \\
  p^{\prime} = p \cap \cons{\neg BOP(m(r_2))}\\
\end{gather*}
For instance, 
in Figure~\ref{fig:example}
we add the constraint $x_7=0$ on
edge $(L7,L8)$ (appearing also on edge $(L8,L9)$). The constraint
$x_7\neq 0$ cannot be expressed as a linear constraint, so it is not
added on edge $(L7,L9)$.


\subsection{Load}

The impact of the immediate load instruction is straightforward:

\begin{align*}
  Update&(LOADI~r_1~c,(p,m,\derefer))=(p^{\prime},m^{\prime},\derefer)\\
        & p^{\prime} = p \cap \cons{x_i = c} \quad m^{\prime} = m[r_1:x_i]\quad x_i=\FV
\end{align*}

Let us now consider the non-immediate load instruction. If the input
state contains a memory address variable that is equivalent to the load
address, then in the output state the value of the destination register
is the value of the memory value mapped to this address:
\begin{align*}
	Update&(LOAD~r_1~r_2, (p, m, \derefer)) =
                (p^{\prime}, m^{\prime}, \derefer) \\
              & p^{\prime} = p \cap \cons{x_i = \derefer(a)} \quad m^{\prime} = m[r_1:x_i] \\
              & a\equiv r_2 \quad\quad\quad\quad x_i=\FV
\end{align*}
Otherwise, the destination register value is undefined:
\begin{align*}
	Update&(LOAD~r_1~r_2, (p, m)) = (p, m^{\prime}) \\
	& m^{\prime} = m[{r_1}:\FV]
\end{align*}

For instance, 
in Figure~\ref{fig:example}
Line~9 we have
$x_5\equiv r_3$ and $\derefer(x_5)=x_6$, so we introduce the constraint
$x_{10}=x_6$ and $m'[R6=x_{10}]$.

\subsection{Store}

Again, we need to consider the impact of aliases. First, we will define
two helper functions. The \emph{Create} operation is used to create a
new memory mapping. It takes as parameters the register containing the
address ($r_d$), the register holding the value to store at this address
($r_a$), and the current abstract state.
\begin{align*}
  Create &(r_d, r_a, (p, m,\derefer)) = (p^{\prime}, m,\derefer')\\
         & p^{\prime} = p \cap \cons{x_i = m(r_d), x_j = m(r_a)} \\
         & \derefer^{\prime} = \derefer[x_i:x_j] \quad x_i=\FV \quad x_j=\FV
\end{align*}
For instance, 
in Figure~\ref{fig:example}
Line~5 creates a new memory
mapping: it introduces the
constraints $x_5=x_3$, $x_6=x_1$ and $\derefer_1(x_5)=x_6$.

The \emph{Replace} operation is used to handle the replacement of the
value of an already mapped memory address, overwriting the previous
value. It takes as parameters the variable representing the memory
address ($a$), the register holding the new value ($r_a$), and the
current abstract state.
\begin{align*}
  Replace &(a, r_a, (p, m, \derefer)) = (p^{\prime}, m, \derefer^{\prime}) \\
          & p^{\prime} = p \cap \cons{x_i = m(r_a)} \\
          & \derefer^{\prime} = \derefer[a:x_i] \quad\quad x_i=\FV
\end{align*}
For instance, 
in Figure~\ref{fig:example}
Line~8 replaces a previous
mapping: it introduces the constraint $x_8=x_2$ and maps $x_5$ to $x_8$
(instead of $x_6$ previously).

Finally, the new abstract state is computed as follows (note that there
is at most one address variable $a$ such that $a\equiv r_1$):
\begin{gather*}
	Update(STORE~r_1~r_2, (p, m, \derefer)) = 
                \begin{cases}
                  Replace(a, r_2, s^{\prime}) & \text{if }a\equiv r_1\\
                  Create(r_1, r_2, s^{\prime}) & \text{otherwise}
                \end{cases}\\
	s^{\prime} = \underset{\{a\in A| a\textsf{ overlaps } r_1)}{\join} (Replace(a, r_2, (p, m,\derefer)),(p, m, \derefer))
\end{gather*}

\subsection{Join and widening}

The procedure for joining two abstract states is detailed in
Algorithm~\ref{alg:join}. Polyhedra are joined using the classic
polyhedra convex hull (line~\ref{algstep:hull}). Concerning register and
memory mappings, we first unify equivalent variables
(lines~\ref{algstep:eq-adr} and~\ref{algstep:eq-register}), so that the
same variable is used to represent the same register or memory location
in both input states. Then, if a memory location or register is bound in
one input state and unbound in the other, it is associated with a fresh
variable in the output state (lines~\ref{algstep:unbound-reg}
and~\ref{algstep:unbound-adr}), meaning that we consider that there are
no constraints concerning it. The widening operation is defined exactly
in the same way, except that $\widen$ is used in place of $\cup$.

\begin{algorithm}
  \caption{Computing $(p,m,\derefer)=\join((p_1,m_1,\derefer_1),(p_2,m_2,\derefer_2))$}
  \label{alg:join}

  \begin{algorithmic}[1]
    \State $p_{1,2}\gets p_1\cap p_2$;
    \State $(p'_1,m'_1,\derefer'_1)\gets (p_1,m_1,\derefer_1)$
    \For{each $(x_1,x_2)\in \varsA{p1}\times\varsA{p2}$}
    \If{$x_1$ is equivalent to $x_2$}
    \State Replace $x_1$ by $x_2$ and
    $\derefer(x_1)$ by $\derefer(x_2)$ in $p^{\prime}_1$, $m^\prime_1$,
    and $\derefer^{\prime}_1$\label{algstep:eq-adr}
    \EndIf
    \EndFor
    \For{Each $r\in Dom(m^{\prime}_1)\cap Dom(m_2)$}
    \State Replace $m^{\prime}_1(r)$ by $m_2(r)$ in $p^{\prime}_1$, $m^\prime_1$, and
    $\derefer^{\prime}_1$ \label{algstep:eq-register}
    \EndFor
    \State $p\gets p^{\prime}_1 \cup p_2$\label{algstep:hull}
    \For{each $r\in Dom(m^{\prime}_1)$}
    \State{$m(r)\gets(m^{\prime}_1(r)=m_2(r))~?~m^{\prime}_1(r)~:~\FV$}\label{algstep:unbound-reg}
    \EndFor
    \For{each $a\in Dom(\derefer^{\prime}_1)$}
    \State{$\derefer(a)\gets(\derefer^{\prime}_1(a)=\derefer_2(a))~?~\derefer^{\prime}_1(a)~:~\FV$}\label{algstep:unbound-adr}
    \EndFor
  \end{algorithmic}
\end{algorithm}

This join operation is illustrated in 
in Figure~\ref{fig:example}
for
$join(e,e')$. Here $s_1$ corresponds to the state of $e$ and $s_2$ to
the state of $e'$. At the unification step, $x_5$ in $s_1$ is equivalent
to $x_5$ in $s_2$, so we substitute $x_6$ for $x_8$ in $\derefer_2$ and
in the polyhedron constraints, so we obtain
$p'_1=(p_1\cap\cons{x_7=0,x_6=x_2})$. Then, the convex hull regroups the
constraints of $p'_1$ and $p_2$ on $x_6$ so we obtain
$x_1\leq x_6\leq x_2$. The convex hull also lifts the constraints on
$x_7$ and $x_8$.

\subsection{Soundness}

The general principle of the proof of soundness of an abstract
interpretation framework is based on a \emph{soundness relation}
$\sigma$, which relates the concrete semantics $c$ of a program $p$ to
its abstract semantics $a$. The abstraction is sound (with respect to
$\sigma$) iff $\sigma(c,a)$ holds for any program
$p$~\cite{cousot1992}.

In our case, the abstract semantics is defined by the abstract states
assigned to the edges of a program. We define the \emph{concrete state}
of a program as the valuation of registers and memory locations. Note
that, due to branching instructions, we may have several possible
valuations for the same edge. Formally, a concrete state is a set of
pairs $(m_c,\derefer_c)$. A \emph{register valuation} $m_c$ maps
registers to their value, while a \emph{memory valuation} $\derefer_c$
maps memory addresses to their value. Valuations are partial
because some registers or addresses may be mapped to no value
(i.e. undefined). The soundness relation is stated in the following
theorem.

\begin{theorem}
  For any program $b$, for any edge $e$ of $b$, let $\mathcal{C}$ be the
  concrete state of $e$ and let $(p,m,\derefer)$ be the abstract state
  of $e$ (at some point of the analysis). Then, all valuations
  $(m_c,\derefer_c)$ of $\mathcal{C}$ satisfy the constraints of
  $p$. This is denoted $\sigma(\mathcal{C},(p,m,\derefer))$.
\end{theorem}

\begin{proof}
  We proceed by induction on the structure of the analysis (i.e. on the
  definition of the update, join and widening operations). The base of the induction
  is straightforward: the initial abstract state has no constraints so
  any concrete state is valid.

Now, we need to prove that, for each operation $\update$, $\widen$,
$\join$, assuming that $\sigma$ holds for the input state(s) of these
operations, it also holds for their output state(s).

\paragraph{Binary operation, Branching, LOADI:} Let us first consider
binary operations. In the concrete semantics, we have:

\begin{align*}
  Update_c & (OP~r_1~r_2~r_3, \mathcal{C}) =\bigcup_{(m_c,\derefer_c)\in \mathcal{C}} (m_c[r_1:OP(m_c(r_2),m_c(r_3))],\derefer_c)
\end{align*}

In the abstract semantics, if $OP$ does not define a linear relation,
then no constraints are added, so $\sigma$ directly holds thanks
to the induction hypothesis. Otherwise, the only new constraint is
$x_i=OP(r_2,r_3)$, with $m(r_1)=x_i$. Considering the induction
hypothesis and the fact that for all pairs $(m'_c,\derefer'_c)$ of the output
state, $m'_c(r_1)=OP(m'_c(r_2),m'_c(r_3))$, the new constraint holds. The
\emph{Branching} and \lstinline!LOADI!  cases are proved with similar
reasoning.

\paragraph{LOAD:} In the concrete semantics, we have:

\begin{align*}
  Update_c & ((LOAD~r_1~r_2, \mathcal{C}) = \bigcup_{(m_c,\derefer_c)\in \mathcal{C}}
             (m_c[r_1:\derefer_c(m_c(r_2))],\derefer_c)
\end{align*}

The case where no address variable is equivalent to $r_2$ is trivial
because it introduces no new constraints. Otherwise, we have
$x_i=\derefer(a)$, with $m[r_1:x_i]$. Since $a$ is equivalent to $r_2$,
the property holds.

\paragraph{STORE:} In the concrete semantics, we have:

\begin{align*}
  Update_c &
           (STORE~r_1~r_2,\mathcal{C})=\bigcup_{(m_c,\derefer_c)\in
             \mathcal{C}}
             (m_c,\derefer_c[(m_c(r_1)):m_c(r_2)])
\end{align*}
Let us first assume that there are no polyhedron variables overlapping
$r_1$. In the $create$ case, the new constraints are
$x_i = m(r_1), x_j = m(r_2)$, with $\derefer[x_i:x_j]$, so $\sigma$
holds. In the $replace$ case, the new constraint is $x_i = m(r_2)$, with
$\derefer[a:x_i]$ and $a$ equivalent to $r_1$, so $\sigma$
holds. Finally, if some variable $a$ overlaps $r_1$, we perform a
$\join$ operation between $p$ and $p$ plus the same constraints as for
the $replace$ case. $a$ overlaps $r_1$ means that either $a\neq r_1$, in
which case we add no constraints ($p^{\prime}=p$), or $a=r_1$, in which
case we add the same constraints as for the replace. The soundness of
the $\join$ operation used here is proved below.

\paragraph{$\join$ and $\widen$:} Concerning the $\join$ operation, in the
concrete semantics we have:
\begin{align*}
  \join_c &
            (\mathcal{C}_1,\mathcal{C}_2)= \mathcal{C}_1\cup\mathcal{C}_2
\end{align*}
Let $(p_1,m_1,\derefer_1)$ and $(p_2,m_2,\derefer_2)$ be the abstract
states corresponding to $\mathcal{C}_1$ and $\mathcal{C}_2$. In the
abstract semantics, we join polyhedra by computing their convex hull. By
definition of the convex hull, and by the induction hypothesis, all the
valuations of $\mathcal{C}_1\cup\mathcal{C}_2$ satisfy the constraints
of $p_1\cup p_2$. Concerning register and memory mappings, the
unification of equivalent variables (lines~\ref{algstep:eq-adr}
and~\ref{algstep:eq-register} of Algorithm~\ref{alg:join}) does not
change constraints. The introduction of fresh variables, when a register
or memory location is unbound in one of the input states
(lines~\ref{algstep:unbound-reg} and~\ref{algstep:unbound-adr}),
effectively relaxes the constraints on it, so the soundness
holds. Soundness of operation $\widen$ holds, from a similar reasoning.

~ \qed
\end{proof}

\subsection{Optimization}
\label{sec:optimizations}

The complexity of our method depends on the number of variables and the
number of constraints created as the analysis progresses. The number of
variables introduced can be easily upper-bounded. For every instruction
in the code, we introduce at most 5 variables: two new variables for
STORE and 3 new variables for the \emph{Merge}. This may seem a lot,
however several optimizations are possible.

The most relevant is the elimination of unused variables from the
polyhedra as the analysis progresses: any variable that is not in $m(.)$
or in $\derefer(.)$ can be safely removed from the polyhedra by
performing a projection on the remaining (used) variables. For example,
the \emph{Merge} operation unifies two existing variables, thus it is
easy to see that after a \emph{Merge} the number of useful variables in
the polyhedron is actually reduced by one. The elimination of unused
variables is implemented in the current version of our tool.

Several other optimizations are possible. An important one is to
do the analysis at the level of functions: when the function
ends, we can safely remove all variables that refer to the
local context of the function. Also, by introducing techniques from
\emph{data structure analysis}, we can significantly reduce the number
of variables that are necessary to investigate the properties of
simple data structures like arrays. We plan to implement such
optimizations as future work.



\section{Loop bounds}

In this section, we show how to apply our abstract interpretation to the
problem of loop bounds estimation. We are interested in two types of
bounds for each loop. The \emph{max} bound specifies the maximum number
of times the loop body will be executed, for each complete execution of
the loop. The \emph{total} bound specifies the total number of times the
loop body will be executed in the whole program. The distinction makes
sense only in the case of nested loops.

\begin{ex}Consider the following code snippet:
\label{ex:tloop}
\begin{lstlisting}[basicstyle=\scriptsize\ttfamily]
int k = 0;
for (int i = 0; i < 10; i++)
  for (int j = 0; j < i; j++)
    k++;
\end{lstlisting}
In this example, the max bound of the inner loop is $9$, while
its total bound is $45$ (that is, $\sum_{n=1..9} n$). The max and total
bound of the outer loop are equal to $10$.
\end{ex}

The general idea of our loop bound estimation is to count loop
iterations using ``virtual'' registers. Before starting the analysis,
for each loop $l_h$ we create two virtual registers $rm_{l_h}$ (for the
max bound) and $rt_{l_h}$ (for the total bound). The constraints on these
registers are updated during the abstract interpretation of the program.


To analyse max bounds, special updates are applied on loop entries and
on loop back-edges (in addition to the updates defined previously in
Section~\ref{sec:interpretation}). When interpreting an entry edge of
loop $l_h$, we assign value $0$ to the corresponding max bound virtual
register:
\begin{align*}
UpdateEntryMax((p, m, d), l_h) = Update(LOADI~rm_{l_h}~0, (p, m, d))
\end{align*}

When interpreting a back-edge of loop $l_h$, we increment the max bound
virtual register (to simplify the presentation, we abusively directly
add a constant value instead of a register content):
\begin{align*}
UpdateIterMax(p, m, d, l_h) = Update(ADD~rm_{l_h}~rm_{l_h}~1, (p, m, d))
\end{align*}

Let $(p_f, m_f, d_f)$ be the final abstract state, i.e. the state
obtained as output of the exit node of the program. Once the analysis is
complete, the max bound of a loop $l_h$ is computed as
$\maxpoly(p_f,m[rm_{l_h}])$.


The analysis of total bounds presented here only considers
\emph{triangular loops}. A triangular loop (see Example~\ref{ex:tloop}
for instance) consists of an inner loop $l_i$ nested inside an outer
loop $l_o$, where the index of the inner loop depends on the index of the
outer loop. The objective of the analysis is to try to establish a
linear relation between them.

First, we define two auxiliary functions. The function
$Relation(p, x_1, x_2)$ tries to find a linear relation between
variables $x_1$ and $x_2$ in $p$:

\[
Relation(p, x_1, x_2) =
\begin{cases}
  (A, B, C) &\text{ if }A.x_1 + B.x_2 \le C \text{ in }\proj{p}{x_1,x_2} \\  
  undefined &\text{ otherwise}
\end{cases}
\]

The function $Total(A, B, C, M)$ computes a sum based on coefficients
$A$, $B$, $C$ (provided they are not undefined) and on a bound $M$:
\begin{align*}
  Total&(A, B, C, M) = \sum_{i=0}^{M - 1}{max\left(0, \left\lfloor {{C - Bi} \over {A}} \right\rfloor\right)}
\end{align*}

Let us now detail the analysis. At the program entry, each total bound
virtual register is set to $0$. The $UpdateIterTotal$ function is applied
on the back-edge of the inner loop. It updates the total bound virtual
register:
\begin{align*}
UpdateIterTotal((p, m, d), l_i) = Update(ADD~rt_{l_i}~rt_{l_i}~1, p, m, d)
\end{align*}

The $UpdateExitTotal$ function is applied on the exit-edge of the inner loop. It
bounds the value of $rt_{l_i}$:
\begin{align*}
	UpdateExitTotal&((p, m, d), l_i, l_o) =  (p^{\prime}, m, d) \\
&p' = p \cap \cons{m[rt_{l_i}] \le t} \\
&(A,B,C) = Relation(p, m[rm_{l_i}], m[rm_{l_o}]) \\
&M = \maxpoly(p,m[rm_{lo}]) \quad \quad t = Total(A, B, C, M)
\end{align*}

Let us consider Example~\ref{ex:tloop}. In the input state
$(p,m,\derefer)$ of the
back-edge of $l_i$ the constraint $rm_{l_o} = rm_{l_i}$ holds, so
$Relation(p, rm_{l_i}, rm_{l_o}])=(1,-1, 0)$. Since
$\max(p,rm_{l_o})=10$, we get
$Total(1, -1, 0, 10)=\sum_{n=1}^{10-1} n= 45$. Hence, we add constraint
$rt_{l_i} \le 45$ to $p$.

Note that the virtual register $rt_{l_i}$ is not actually used to
compute the total loop bound. It can however be used to analyse code
executed after the nested loops. For instance, in
Example~\ref{ex:tloop}, we obtain the constraint $k\leq 45$, should $k$
appear somewhere later in the program.


\section{Experimental results}

Our methodology is implemented in a prototype called \ourtool,
as a plugin of OTAWA (version 2.0), an open source WCET computation
tool~\cite{otawa}. \ourtool\ relies on OTAWA for CFG construction and
manipulation, and on PPL~\cite{bagnara2008parma} for polyhedra
operations. The analyses have been executed on a PC with an Intel core
i5 3470 at 3.2 Ghz, with 8 Gb of RAM. Every benchmark has been
compiled with ARM crosstool-NG 1.20.0 (gcc version 4.9.1) with
\textsf{-O1} optimization level.


First, we report the results of our experiments on the M\"alardalen
benchmarks~\cite{gustafsson2010malardalen} in Table~\ref{tab:xp}. We exclude benchmarks that are not supported by OTAWA,
mainly due to floating point operations or indirect branching
(e.g. \lstinline!switch!). We compare \ourtool\ with SWEET~\cite{lisper2014sweet}, Pagai~\cite{henry2012pagai} and
\orange~\cite{bonenfant2008orange}. For each benchmark, we report: the number of lines of code (in the C source), the
total number of loops, the number of loops that are correctly bounded by
each tool, and the computation time. We do not report the computation
time for SWEET because we only had access to it through an online
applet. For \orange, computation time is below the measurement
resolution (10ms), except for $edn$, where it reaches 50ms. 

The execution time of \ourtool\ is typically higher than that of Pagai because
we introduce more variables and constraints. We believe we can
reduce the gap with additional optimization of the method and of the
code, however \ourtool\ will probably remain more costly because it
works at a lower level of abstraction.

Concerning loop bounds 
there are two benchmarks for which 
\ourtool\ did not find any loop bound: for bench \textsf{edn}, \ourtool\
is unable to prove that there is no array out-of-bound accesses which
could potentially overwrite the loop index. For
\textsf{janne\_complex}, the difficulty is that it contains complex
loop index updates inside a \lstinline!if-then-else!. Furthermore,
note that Pagai does not compute total loop bounds.


\begin{table}
  \centering
{\small  \begin{tabular}{|c||c|c|c|c|c|c|c|c|}
    \hline
	\multicolumn{3}{|c}{} & \multicolumn{4}{|c}{Loops Correctly Bounded} &  \multicolumn{2}{|c|}{Time (ms)}  \\
	\hline
	\emph{Bench}            & \emph{LoC} & \emph{Loops} &\emph{\ourtool} & \emph{SWEET} & \emph{Pagai} & \emph{\orange} & \emph{\ourtool} & \emph{Pagai} \\
    \hline
crc				& 16	& 1		& 1		& 1		& 1		& 1		& 150			& 40\\
fibcall			& 22	& 1		& 1		& 1		& 1		& 1		& 230			& 50	\\
janne\_complex	& 26	& 2		& 1		& 2		& 1		& 1		& 870			& 140	\\
expint			& 56	& 3		& 3		& 2		& 3		& 3		& 850			& 9140	\\
matmult		    & 84	& 5		& 5		& 5		& 5		& 5		& 3640			& 1380	 \\
fdct			& 149	& 2		& 2		& 2		& 2		& 2		& 12450			& 2150 \\
jfdctint	    & 165	& 3		& 3		& 3		& 3		& 3		& 10920			& 1960 \\
fir				& 189	& 2		& 2		& 2		& 2		& 1		& 11630			& 390	\\
edn				& 198	& 12	& 12    & 12    & 9		& 12    & 25190			& 15660 \\
ns		        & 414	& 4		& 4		& 4		& 4		& 4		& 1700			& 380 \\
    \hline
  \end{tabular}}
\caption{Benchmark results}
\label{tab:xp}
\end{table}

We further illustrate the differences between tool capabilities on the
two examples of Figure~\ref{fig:loop-ex}. For both examples, \ourtool\
provides the correct max and total loop bounds. For the example
\textsf{foo1}, \orange\ fails to compute the max and total bounds of
the inner loop, because it does not notice that
\lstinline!i-x<10!. For \textsf{foo2}, Pagai does not find the max
loop bound (the loop is considered unbounded), because it does not
infer that \lstinline!*ptr=&bound! when executing instruction
\lstinline!*ptr=15!.

\begin{figure}
  \begin{minipage}{.48\linewidth}
    \begin{lstlisting}[basicstyle=\scriptsize\ttfamily]
foo1(int x) { 
  int i = 0, j = 0, k = 0;
  for (i = x; i < (x + 10); i++)
    for (j = 0; j < (i - x); j++);
}
    \end{lstlisting}
  \end{minipage}
  \hspace{.5cm}
  \begin{minipage}{.45\linewidth}
    \begin{lstlisting}[basicstyle=\scriptsize\ttfamily]
foo2() {
  int i, bound = 10;
  int *ptr = &bound;
  ptr++; ptr--; *ptr = 15; k = 0;
  for (i = 0; i < bound; i++);
}      
    \end{lstlisting}
  \end{minipage}
  \caption{Loop examples}
  \label{fig:loop-ex}
\end{figure}


\section{Conclusion}
In this paper we propose a novel technique for performing abstract
interpretation of binary code using polyhedra. Our method consists in
adding new variables to the polyhedra as the analysis progresses, and
maintaining a correspondence with registers and memory addresses.
Thanks to the relational properties of polyhedra, our technique
naturally provides information on pointer aliasing when compared to
other techniques based on non-relational domains. While the complexity
of our method is currently still relatively high, we believe that
there is room for improvement: we are planning to apply some
well-known techniques from static analysis to reduce the number of
variables to analyse.



\bibliographystyle{splncs}

\newpage
\bibliography{biblio}

\end{document}